\documentclass[conference]{IEEEtran}
\IEEEoverridecommandlockouts

\usepackage{color}

\usepackage{graphicx,tabularx,array,amsmath,amsthm,thmtools}

\usepackage{mathtools}
\usepackage{caption}
\usepackage{amsfonts}
\usepackage{bbm}
\usepackage{makecell}
\usepackage{multirow}
 \usepackage{amssymb}
\usepackage{txfonts}
\usepackage[T1]{fontenc}
\usepackage{tikz}
\usepackage[scr=dutchcal]{mathalfa}
\let\mathscr\mathbscr

\newtheorem{Theorem}{Theorem}
\newtheorem{Proposition}{Proposition}
\newtheorem{Lemma}{Lemma}
\newtheorem{Corollary}{Corollary}
\newtheorem{Example}{Example}
\newtheorem{Remark}{Remark}
\newtheorem{Definition}{Definition}

\ifCLASSINFOpdf
\else
\fi
\begin{document}
%
\title{On the Correlation between Boolean Functions of Sequences of Random Variables}

\author{\IEEEauthorblockN{Farhad Shirani Chaharsooghi}
\IEEEauthorblockA{Electrical Engineering and\\Computer Science\\
University of Michigan\\
Ann Arbor, Michigan, 48105\\
Email: fshirani@umich.edu }
\and
\IEEEauthorblockN{S. Sandeep Pradhan}
\IEEEauthorblockA{Electrical Engineering and\\Computer Science\\
University of Michigan\\
Ann Arbor, Michigan, 48105\\
Email: pradhanv@umich.edu}}


%


\maketitle

\begin{abstract}
In this paper, we establish a new inequality tying together the effective length and the maximum correlation between the outputs of an arbitrary pair of Boolean functions which operate on two sequences of correlated random variables. We derive a new upper-bound on the correlation between the outputs of these functions. The upper-bound is useful in various disciplines which deal with common-information. We build upon Witsenhausen's \cite{ComInf2} bound on maximum-correlation.  The previous upper-bound did not take the effective length of the Boolean functions into account.   
One possible application of the new bound is to characterize the communication-cooperation tradeoff in multi-terminal communications. In this problem, there are lower-bounds on the effective length of the Boolean functions due to the rate-distortion constraints in the problem, as well as lower bounds on the output correlation at different nodes due to the multi-terminal nature of the problem.
\end{abstract}


%
\IEEEpeerreviewmaketitle

\section{Introduction}
{A}{} fundamental problem of broad theoretical and practical interest is to characterize the maximum correlation between the outputs of a pair of functions of random sequences.  Consider the two distributed agents shown in Figure \ref{fig:agents}.
A pair of correlated discrete memoryless sources (DMS) are fed to the two agents. These agents are to each make a binary decision. The goal of the problem is to maximize the correlation between the outputs of these agents subject to specific constraints on the decision functions. The study of this setup has had impact on a variety of disciplines, for instance, by taking the agents to be two encoders in the distributed source coding problem \cite{FinLen,arxiv2}, or two transmitters in the interference channel problem \cite{arxiv2}, or Alice and Bob in a secret key-generation problem \cite{security2, security3}, or two agents in a distributed control problem \cite{control}. 

A special case of the problem is the study of common-information (CI) generated by the two agents. As an example, consider two encoders in a Slepian-Wolf (SW) setup. Let $U_1,U_2$, and $V$ be independent, non-constant binary random variables.
\begin{figure}[!t]
\centering
\includegraphics[height=1.2in]{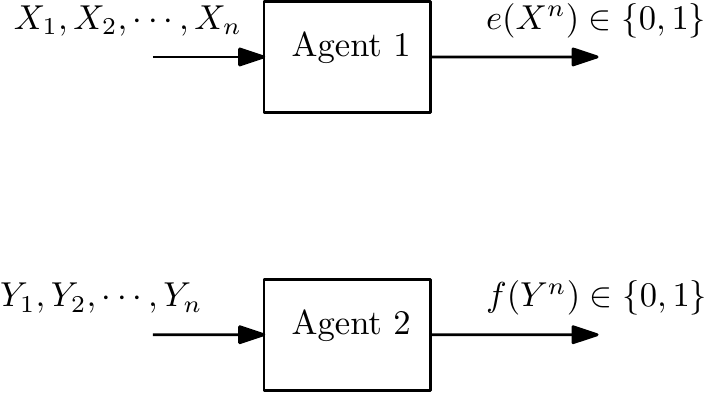}
 \caption{Correlated Boolean decision functions.}
\label{fig:agents}
\end{figure}
Then, an encoder observing the DMS $X=(V,U_1)$, and an encoder observing $Y=(V,U_2)$ agree on the value of $V$ with probability one. The random variable $V$ is called the CI observed by the two encoders. These encoders require a sum-rate equal to $H(V)+H(U_1)+H(U_2)$ to transmit the source to the decoder. This gives a reduction in rate equal to the entropy of $V$, compared to the transmission of the sources over independent point-to-point channels. The gain in performance is directly related to the entropy of the CI. So, it is desirable to maximize the entropy of the CI between the encoders.
 
  In \cite{ComInf1}, the authors investigated multi-letterization as a method for increasing the CI. They showed that multi-letterization does not lead to an increase in the CI. More precisely, they prove the following statement: 
  
{ \textit{Let $X$ and $Y$ be two sequences of DMSs. Let $f_{n}(X^n)$ and $g_{n}(Y^n)$ be two sequences of functions which converge to one another in probability. Then, the normalized entropies $\frac{1}{n}H(f_{n}(X^n))$, and $\frac{1}{n}H(g_{n}(Y^n))$ are less than or equal to the entropy of the CI between $X$ and $Y$ for large $n$. }}
 
 A stronger version of the result was proved by Witsenhausen \cite{ComInf2}, where maximum correlation between the outputs is upper-bounded subject to the following restrictions on the decision functions:
\\\textit{ 1) The entropy of the binary output is fixed.
 \\2) The agents cooperate with each other.
} 
 
 It was shown that maximum correlation is achieved if both users output a single element of the string without further processing (e.g. each user outputs the first element of its corresponding string). This was used to conclude that common-information can not be induced by multi-letterization. While, the result was used extensively in a variety of areas such as information theory, security, and control \cite{security2,security3, control}, in many problems, there are additional constraints on the set of admissible decision functions. For example, one can consider constraints on the `effective length' of the decision functions. This is a valid assumption, for instance, in the case of communication systems, the users have lower-bounds on their effective lengths due to the rate-distortion requirements in the problem \cite{arxiv2}. 
 
 In this paper, the problem under these additional constraints is considered.  A new upper-bound on the correlation between the outputs of arbitrary pairs of Boolean functions is derived. The bound is presented as a function of the dependency spectrum of the Boolean functions.  This is done in several steps. 
First, the effective length of an additive Boolean function is defined. Then, we use a method similar to \cite{ComInf2}, and map the Boolean functions to the set of real-valued functions. Using tools in real analysis, we find an additive decomposition of these functions. The decomposition components have well-defined effective lengths. Using the decomposition we find the dependency spectrum of the Boolean function. The dependency spectrum is a generalization of the effective length and is defined for non-additive Boolean functions. Lastly, we use the dependency spectrum to derive the new upper-bound.

The rest of the paper is organized as follows: Section \ref{sec:not} presents the notation used in the paper. Section \ref{sec:eff} develops useful mathematical machinery to analyze Boolean function. Section \ref{sec:corr} contains the main result of the paper. Finally, Section \ref{sec:con} concludes the paper.  
  
 

\section{Notation}\label{sec:not}
In this section, we introduce the notation used in this paper. We represent random variables by capital letters such as $X, U$. Sets are denoted by calligraphic letters such as $\mathcal{X}, \mathcal{U}$.  Particularly, the set of natural numbers and real numbers are shown by $\mathbb{N}$, and $\mathbb{R}$, respectively. 
For random variables, the $n$-length vector $(X_1,X_2,\cdots,X_n), X_i\in \mathcal{X}$ is denoted by $X^n\in \mathcal{X}^n$. 
 The binary string $(i_1,i_2,\cdots,i_n), i_j\in \{0,1\}$ is written as $\mathbf{i}$.
 The vector of random variables $(X_{j_1},X_{j_2},\cdots, X_{j_k}), j_i\in [1,n], j_i\neq j_k$, is denoted by $X_{\mathbf{i}}$, where $i_{j_l}=1, \forall l\in [1,k]$. For example, take $n=3$, the vector $(X_1,X_3)$ is denoted by $X_{101}$, and the vector $(X_1,X_2)$ by $X_{110}$.  
 For two binary strings $\mathbf{i},\mathbf{j}$, we write $\mathbf{i}<\mathbf{j}$ if and only if $i_k<j_k, \forall k\in[1,n]$. For a binary string $\mathbf{i}$ we define $N_{\mathbf{i}}\triangleq w_H(\mathbf{i})$, where $w_H$ denotes the Hamming weight. Lastly, the vector $\sim \mathbf{i}$ is the element-wise complement of $\mathbf{i}$.

\section{The \textit{Dependency Spectrum} of a Function}\label{sec:eff}
In this section, we study the correlation between the output of a Boolean function with subsets of the input. Particularly, we are interested in the answers to questions such as `How strongly does the first element $X_1$ affect the output of $e(X^n)$?' `Is this effect amplified when we take $X_2$ into account as well?' `Is there a subset of random variables that (almost) determines the value of the output?'. We formulate these questions in mathematical terms, and find a characterization of the dependency spectrum of a Boolean function. The dependency spectrum is a vector which captures the correlation between different subsets of the input elements with each element of the output.  As an intermediate step, we define the effective length of an additive Boolean function below:
\begin{Definition}
 For a Boolean function $e:\{0,1\}^n\to \{0,1\}$ defined by $e(X^n)=\sum_{i\in \mathsf{J}}X_i, \mathsf{J}\subset [1,n] $, where the addition operator is the binary addition, the effective length is defined as the cardinality of the set $\mathsf{J}$.
\end{Definition}
   For a general Boolean function (e.g. non-additive), we find a decomposition of ${e}$ into a set of functions ${e}_{\mathbf{i}}, \mathbf{i}\in \{0,1\}^n$ whose effective length is well-defined. First, we provide a mapping from the set of Boolean functions to the set of real functions. This allows us to use the tools available in real analysis to analyze these functions.
Fix a discrete memoryless source $X$, and a Boolean function defined by ${e}:\{0,1\}^n\to \{0,1\}$.  Let $P\left(e(X^n)=1\right)=q$. The real-valued function corresponding to $e$ is represented by $\tilde{e}$, and is defined as follows:
\begin{align}
\tilde{e}(X^n)=    \begin{cases}
      1-q, & \qquad  e(X^n)=1, \\
      -q. & \qquad\text{otherwise}.
    \end{cases}
\end{align}
\begin{Remark}
 Note that $\tilde{e}$ has zero mean and variance $q(1-q)$.
\label{Rem:exp_0}
\end{Remark}
The random variable $\tilde{e}(X^n)$ has finite variance on the probability space $(\mathcal{X}^n, 2^{\mathcal{X}^n}, P_{X^n})$. The set of all such functions is denoted by $\mathcal{H}_{X,n}$. More precisely,  we define $\mathcal{H}_{X,n}\triangleq L_2(\mathcal{X}^n, 2^{\mathcal{X}^n}, P_{X^n})$ as the separable Hilbert space of all measurable functions $\tilde{h}:\mathcal{X}^n\to \mathbb{R}$.  Since X is a DMS, the isomorphy relation 
 \begin{equation}
 \mathcal{H}_{X,n}= \mathcal{H}_{X,1}\otimes \mathcal{H}_{X,1}\cdots \otimes \mathcal{H}_{X,1}
\label{eq:Hil_Dec1}
 \end{equation}
 holds \cite{Reed_and_Simon}, where $\otimes$ indicates the tensor product. 
\begin{Example}
Let n=1. The Hilbert space $\mathcal{H}_{X,1}$ is the space of all measurable functions $\tilde{h}:\mathcal{X}\to \mathbb{R}$. The space is spanned by the two linearly independent functions  $\tilde{h}_1(X)=\mathbbm{1}(X)$ and $\tilde{h}_2(X)=\mathbbm{1}(\bar{X})$, where $\bar{X}=X\oplus 1$. We conclude that the space is two-dimensional.
\label{Ex:ex1}  
\end{Example}

 \begin{Remark}
  The tensor operation in $\mathcal{H}_{X,n}$ is real multiplication (i.e. $f_1, f_2\in \mathcal{H}_{X,1}: f_1(X_1)\otimes f_2(X_2)\triangleq f_1(X_1)f_2(X_2)$). Let $\{f_i(X)|i\in [1,d]\}$ be a basis for  $\mathcal{H}_{X,1}$, then a basis for  $\mathcal{H}_{X,n}$ would be the set of all the real multiplications of these basis elements: $\{\Pi_{j\in [1,n]}f_{i_j}(X_j), i_j\in [1,d]\}$. 
\end{Remark}

 Example \ref{Ex:ex1} gives a decomposition of the space $\mathcal{H}_{X,1}$. Next, we introduce another decomposition of $\mathcal{H}_{X,1}$ which turns out to be very useful. Let $\mathcal{I}_{X,1}$ be the subset of all measurable functions of $X$ which have 0 mean, and let $\gamma_{X,1}$ be the set of constant real functions of $X$.  We argue that  $\mathcal{H}_{X,1}=\mathcal{I}_{X,1}\oplus \gamma_{X,1}$ gives a decomposition of $\mathcal{H}_{X,1}$. 
 $\mathcal{I}_{X,1}$ and  $\gamma_{X,1}$  are linear subspaces of $\mathcal{H}_{X,1}$. $\mathcal{I}_{X,1}$ is the null space of the linear functional which takes an arbitrary function $\tilde{f}\in \mathcal{H}_{X,1}$ to its expected value $\mathbb{E}_{X}(\tilde{f})$. The null space of any non-zero linear functional is a hyper-space in $\mathcal{H}_{X,1}$. So, $\mathcal{I}_{X,1}$ is a one-dimensional subspace of $\mathcal{H}_{X,1}$. From Remark \ref{Rem:exp_0}, $\tilde{e}_1\in \mathcal{I}_{X,1}$. We conclude that any element of $\mathcal{I}_{X,1}$  can be written as $c\tilde{e}_1(X^n), c\in \mathbb{R}$. 
 $\gamma_{X,1}$ is also one dimensional. It is spanned by the function $\tilde{g}(X)=1$.
   Consider an arbitrary element $\tilde{f}\in \mathcal{H}_{X,1}$. One can write $\tilde{f}= \tilde{f}_1+\tilde{f}_2$ where $\tilde{f}_1=\tilde{f}-\mathbb{E}_{X}(\tilde{f})\in \mathcal{I}_{X,1}$, and $\tilde{f}_2= \mathbb{E}_{X}(\tilde{f})\in  \gamma_{X,1}$.   \label{Ex:one_dim}
Replacing $\mathcal{H}_{X,1}$ with $\mathcal{I}_{X,1}\oplus \gamma_{X,1}$ in \eqref{eq:Hil_Dec1}, we have:
 \begin{align}
 \mathcal{H}_{X,n}&=\otimes_{i=1}^n \mathcal{H}_{X,1}=\otimes_{i=1}^n (\mathcal{I}_{X,1}\oplus \gamma_{X,1})\nonumber
 \\&\stackrel{(a)}{=} \oplus_{\mathbf{i}\in \{0,1\}^n} (\mathcal{G}_{i_1}\otimes \mathcal{G}_{i_2}\otimes\dotsb \otimes \mathcal{G}_{i_n}),
 \label{eq:Hil_Dec2}
\end{align}
where
\begin{align*}
 \mathcal{G}_j=
\begin{cases}
 \gamma_{X,1}  \qquad \ j=0,\\
 \mathcal{I}_{X,1} \qquad \ j=1,
\end{cases}
\end{align*}
and, in (a), we have used the distributive property of tensor products over direct sums. 
\begin{Remark}
 Equation \eqref{eq:Hil_Dec2}, can be interpreted as follows: for any $\tilde{e}\in \mathcal{H}_{X,n},  n\in \mathbb{N}$, we can find a decomposition $\tilde{e}=\sum_{\mathbf{i}}\tilde{e}_{\mathbf{i}}$, where $\tilde{e}_{\mathbf{i}}\in \mathcal{G}_{i_1}\otimes \mathcal{G}_{i_2}\otimes\dotsb \otimes \mathcal{G}_{i_n}$. $\tilde{e}_{\mathbf{i}}$ can be viewed as the component of $\tilde{e}$ which is only a function of $\{X_{i_j}|i_j=1\}$. In this sense, the collection $\{\tilde{e}_{\mathbf{i}}|\sum_{j\in[1,n]}i_j=k\}$, is the set of components of $\tilde{e}$ whose effective length is $k$.
\label{Rem:Dec}
\end{Remark}
In order clarify the notation, we provide the following example:
\begin{Example}
 Let $X$ be a binary symmetric source, and let $e(X_1,X_2)=X_1\wedge X_2$ be the binary `and' function. The corresponding real function is:
 \begin{align*}
 \tilde{e}(X_1,X_2)=
\begin{cases}
 -\frac{1}{4}  \qquad \ (X_1,X_2)\neq (1,1),\\
\frac{3}{4} \qquad \ (X_1,X_2)=(1,1).
\end{cases}
\end{align*} 
Lagrange interpolation gives $\tilde{e}=X_1X_2-\frac{1}{4}$. The decomposition is given by:
\begin{align*}
 &\tilde{e}_{1,1}= (X_1-\frac{1}{2})(X_2-\frac{1}{2}), \tilde{e}_{1,0}=\frac{1}{2}(X_1-\frac{1}{2}), 
 \\&\tilde{e}_{0,1}= \frac{1}{2}(X_2-\frac{1}{2}),\tilde{e}_{0,0}=0.
\end{align*} 
The variances of these functions are given below:
\begin{align*}
 &Var(\tilde{e})=\frac{3}{16}, Var(\tilde{e}_{0,1})=Var(\tilde{e}_{1,0})=Var(\tilde{e}_{1,1})=\frac{1}{16}. 
\end{align*}
As we shall see in the next section, these variances play a major role in determining the correlation preserving properties of $\tilde{e}$. The vector whose elements include these variances is called the dependency spectrum of $e$.
In the perspective of the effective length, the function $\tilde{e}$ has $\frac{2}{3}$ of its variance distributed between $\tilde{e}_{0,1}$, and $\tilde{e}_{1,0}$ which have effective length one, and $\frac{1}{3}$ of the variance is on $\tilde{e}_{1,1}$ which is has effective length two. 
 
\end{Example}
Similar to the above examples, for arbitrary $\tilde{e}\in \mathcal{H}_{X,n},  n\in \mathbb{N}$, we find a decomposition $\tilde{e}=\sum_{\mathbf{i}}\tilde{e}_{\mathbf{i}}$, where $\tilde{e}_{\mathbf{i}}\in \mathcal{G}_{i_1}\otimes \mathcal{G}_{i_2}\otimes\dotsb \otimes \mathcal{G}_{i_n}$. We characterize $\tilde{e}_{\mathbf{i}}$ in terms of products of the basis elements of $\otimes_{j\in[1,n]}\mathcal{G}_{i_j}$ using the following result in linear algebra:
\begin{Lemma}[\cite{Reed_and_Simon}]
 Let $\mathcal{H}_{i},i \in[1,n]$ be vector spaces over a field $F$. Also, let $\mathcal{B}_{i}=\{v_{i,j}|j\in [1,d_i]\}$ be the basis for $\mathcal{H}_i$ where $d_i$ is the dimension of $\mathcal{H}_i$. Then, any element $v\in \otimes_{i\in [1,n]}\mathcal{H}_i$ can be written as $v=\sum_{j_1\in [1,d_1]}\sum_{j_2\in [1,d_2]}\cdots \sum_{j_n\in [1,d_n]} c_{j^n} v_{j_1}\otimes v_{j_2}\cdots \otimes v_{j_n}$.
 \label{Lem:tensor_dec}
\end{Lemma}

Since $\mathcal{G}_{i_j}$'s, $j\in [1,n]$ take values from the set $\{\mathcal{I}_{X,1}, \gamma_{X,1}\}$, they are all one-dimensional. For the binary source $X$ with $P(X=1)=q$, define $\tilde{h}$ as:
\begin{align}
\tilde{h}(X)=    \begin{cases}
      1-q, & \text{if } X=1, \\
      -q. & \text{if } X=0.
    \end{cases}
    \label{eq:basis} 
\end{align}
Then, the single element set $\{\tilde{h}(X)\}$ is a basis for $\mathcal{I}_{X,1}$. Also, the function $\tilde{h}(X)=1$ spans $\gamma_{X,1}$. So, using Lemma \ref{Lem:tensor_dec}, $\tilde{e}_{\mathbf{i}}(X^n)= c_{\mathbf{i}}\prod_{t:i_t=1}\tilde{h}(X_{{t}}), c_i\in \mathbb{R}$. We are interested in the variance of $\tilde{e}_{\mathbf{i}}$'s. In the next proposition, we show that the $\tilde{e}_{\mathbf{i}}$'s are uncorrelated and we derive an expression for the variance of $\tilde{e}_{\mathbf{i}}$.
\begin{Proposition}
\label{pr:partfun}
Define $\mathbf{P}_{\mathbf{i}}$ as the variance of $\tilde{e}_{\mathbf{i}}$. The following hold:
 \\ 1) $\mathbb{E}(\tilde{e}_{\mathbf{i}}\tilde{e}_{\mathbf{j}})=0, \mathbf{i}\neq \mathbf{j}$, in other words $\tilde{e}_{\mathbf{i}}$'s are uncorrelated.
 \\ 2)  $\mathbf{P}_{\mathbf{i}}=\mathbb{E}(\tilde{e}_{\mathbf{i}}^2)=c_{\mathbf{i}}^2(q(1-q))^{w_{H}(\mathbf{i})}.$
\end{Proposition}
\begin{proof}
 1) follows by direct calculation. 2) holds from the independence of $X_i$'s.
\end{proof}

Next, we find the characterization for $\tilde{e}_{\mathbf{i}}$.

\begin{Lemma}
$\tilde{e}_{\mathbf{i}}=\mathbb{E}_{X^n|X_{\mathbf{i}}}(\tilde{e}|X_{\mathbf{i}})-\sum_{\mathbf{j}< \mathbf{i}} \tilde{e}_{\mathbf{j}}$ gives the unique orthogonal decomposition of $\tilde{e}$ into the Hilbert spaces $\mathcal{G}_{i_1}\otimes \mathcal{G}_{i_2}\cdots\otimes \mathcal{G}_{i_n}, \mathbf{i}\in \{0,1\}^n$.
\label{Lem:unique}
\end{Lemma}
\begin{IEEEproof}
 Please refer to the Appendix.
\end{IEEEproof}

The following example clarifies the notation used in Lemma \ref{Lem:unique}.
\begin{Example}
 Consider the case where $n=2$. We have the following decomposition of $\mathcal{H}_{X,2}$:
\begin{align}
& \mathcal{H}_{X,2}=(\mathcal{I}_{X,1} \otimes \mathcal{I}_{X,1}) \oplus\nonumber \\
&(\mathcal{I}_{X,1} \otimes \gamma_{X,1})\oplus (\gamma_{X,1}\otimes\mathcal{I}_{X,1}) \oplus (\gamma_{X,1}\otimes\gamma_{X,1}).
 \label{eq:2dimdec}
 \end{align}
Let $\tilde{e}(X_1,X_2)$ be an arbitrary function in $\mathcal{H}_{X,2}$. The unique decomposition of $\tilde{e}$ in the form given in \eqref{eq:2dimdec} is as follows:
\begin{align*}
 \tilde{e}&= \tilde{e}_{1,1}+\tilde{e}_{1,0}+\tilde{e}_{0,1}+\tilde{e}_{0,0},\\ 
 &\tilde{e}_{1,1}= \tilde{e}-\mathbb{E}_{X_2|X_1}(\tilde{e}|X_1)-\mathbb{E}_{X_1|X_2}(\tilde{e}|X_2)+\mathbb{E}_{X_1,X_2}(\tilde{e})
 \\& \tilde{e}_{1,0}=\mathbb{E}_{X_2|X_1}(\tilde{e}|X_1)-\mathbb{E}_{X_1,X_2}(\tilde{e}),
\\& \tilde{e}_{0,1}= \mathbb{E}_{X_1|X_2}(\tilde{e}|X_2)-\mathbb{E}_{X_1,X_2}(\tilde{e}),
\\&\tilde{e}_{0,0}=\mathbb{E}_{X_1,X_2}(\tilde{e}).
\end{align*}
It is straightforward to show that each of the $\tilde{e}_{i,j}$'s, $i,j\in \{0,1\}$, belong to their corresponding subspaces. For instance, $ \tilde{e}_{0,1}$ is constant in $X_1$, and is a $0$ mean function of $X_2$ (i.e. $\mathbb{E}_{X_2}\left(\tilde{e}_{0,1}(x_1,X_2)\right)=0, x_1\in \{0,1\}$), so  $\tilde{e}_{0,1}\in \gamma_{X,1}\otimes\mathcal{I}_{X,1}$.
\end{Example}
The following proposition describes some of the properties of $\tilde{e}_{\mathbf{i}}$ which were derived in the proof of Lemma \ref{Lem:unique}:
\begin{Proposition}
\label{prop:belong2}
The following hold:
\nonumber\\1) $\forall\mathbf{i},\mathbb{E}_{X^n}(\tilde{e}_{\mathbf{i}})$=0.\\
 2) $\forall \mathbf{i}\leq \mathbf{k}$, we have $\mathbb{E}_{X^n|X_{\mathbf{j}}}(\tilde{e}_{\mathbf{i}}|X_{\mathbf{k}})=\tilde{e}_{\mathbf{i}}$.\\
 3) $\mathbb{E}_{X^n}(\tilde{e}_{\mathbf{i}}\tilde{e}_{\mathbf{k}})=0$, for $\mathbf{i}\neq \mathbf{k}$.\\
 4) $\forall \mathbf{k}\leq \mathbf{i}: \mathbb{E}_{X^n|X_{\mathbf{k}}}(\tilde{e}_{\mathbf{i}}|X_{\mathbf{k}})=0.$
 \end{Proposition}

Lastly, we derive an expression for $\mathbf{P}_{\mathbf{i}}$:
\begin{Lemma}
For arbitrary $e:\{0,1\}^n\to \{0,1\}$, let $\tilde{e}$ be the corresponding real function, and let $\tilde{e}=\sum_{\mathbf{i}}\tilde{e}_{\mathbf{i}}$ be the decomposition in the form of Equation \eqref{eq:Hil_Dec2}. The variance of each component in the decomposition is given by the following recursive formula $\mathbf{P}_{\mathbf{i}} =\mathbb{E}_{X_{\mathbf{i}}}(\mathbb{E}_{X^n|X_{\mathbf{i}}}^2(\tilde{e}|X_{\mathbf{i}}))-\sum_{\mathbf{j}< \mathbf{i}}\mathbf{P}_{\mathbf{j}}, \forall \mathbf{i}\in \mathbb{F}_2^n$, where $\mathbf{P}_{\underline{0}}\triangleq 0$. 
\label{Lem:power}
\end{Lemma}
\begin{IEEEproof}
 \begin{align*}
 \mathbf{P}_{\mathbf{i}}
 &=Var_{X_{\mathbf{i}}}(\tilde{e}_{\mathbf{i}}(X^n))
 = \mathbb{E}_{X_{\mathbf{i}}}(\tilde{e}^2_{\mathbf{i}}(X^n))-\mathbb{E}_{X_{\mathbf{i}}}^2(\tilde{e}_{\mathbf{i}}(X^n))\\
& \stackrel{(a)}{=} \mathbb{E}_{X_{\mathbf{i}}}\left(\left(\mathbb{E}_{X^n|X_{\mathbf{i}}}(\tilde{e}|X_{\mathbf{i}})-\sum_{\mathbf{j}< \mathbf{i}} \tilde{e}_{\mathbf{j}}\right)^2\right)-0
\\&=\mathbb{E}_{X_{\mathbf{i}}}\left(\mathbb{E}_{X^n|X_{\mathbf{i}}}^2(\tilde{e}|X_{\mathbf{i}})\right)-2\sum_{\mathbf{j}<\mathbf{i}}\mathbb{E}_{X_{\mathbf{i}}}\left(\mathbb{E}_{X^n|X_{\mathbf{i}}}(\tilde{e}|X_{\mathbf{i}})\tilde{e}_{\mathbf{j}}\right)+\mathbb{E}_{X_{\mathbf{i}}}((\sum_{\mathbf{j}< \mathbf{i}} \tilde{e}_{\mathbf{j}})^2)
\\&\stackrel{(b)}=\mathbb{E}_{X_{\mathbf{i}}}\left(\mathbb{E}_{X^n|X_{\mathbf{i}}}^2(\tilde{e}|X_{\mathbf{i}})\right)-2\sum_{\mathbf{j}<\mathbf{i}}\mathbb{E}_{X_{\mathbf{i}}}\left(\mathbb{E}_{X^n|X_{\mathbf{i}}}(\sum_{\mathbf{l}}\tilde{e}_{\mathbf{l}}|X_{\mathbf{i}})\tilde{e}_{\mathbf{j}}\right)
\\&+\mathbb{E}_{X_{\mathbf{i}}}((\sum_{\mathbf{j}< \mathbf{i}} \tilde{e}_{\mathbf{j}})^2)
\\&\stackrel{(c)}{=}\mathbb{E}_{X_{\mathbf{i}}}\left(\mathbb{E}_{X^n|X_{\mathbf{i}}}^2(\tilde{e}|X_{\mathbf{i}})\right)-2\sum_{\mathbf{j}<\mathbf{i}}\mathbb{E}_{X_{\mathbf{i}}}\left(\sum_{\mathbf{l}}\mathbb{E}_{X^n|X_{\mathbf{i}}}(\tilde{e}_{\mathbf{l}}|X_{\mathbf{i}})\tilde{e}_{\mathbf{j}}\right)
\\&+\mathbb{E}_{X_{\mathbf{i}}}((\sum_{\mathbf{j}< \mathbf{i}} \tilde{e}_{\mathbf{j}})^2)
\\&\stackrel{(d)}{=}\mathbb{E}_{X_{\mathbf{i}}}\left(\mathbb{E}_{X^n|X_{\mathbf{i}}}^2(\tilde{e}|X_{\mathbf{i}})\right)-2\sum_{\mathbf{j}<\mathbf{i}}\mathbb{E}_{X_{\mathbf{i}}}\left(\sum_{\mathbf{l}}\mathbbm{1}(\mathbf{l}\leq \mathbf{i})\mathbb{E}_{X^n|X_{\mathbf{i}}}(\tilde{e}_{\mathbf{l}}|X_{\mathbf{i}})\tilde{e}_{\mathbf{j}}\right)
\\&+\mathbb{E}_{X_{\mathbf{i}}}((\sum_{\mathbf{j}< \mathbf{i}} \tilde{e}_{\mathbf{j}})^2)
\\&\stackrel{(e)}{=}\mathbb{E}_{X_{\mathbf{i}}}\left(\mathbb{E}_{X^n|X_{\mathbf{i}}}^2(\tilde{e}|X_{\mathbf{i}})\right)-2\sum_{\mathbf{j}<\mathbf{i}}\mathbb{E}_{X_{\mathbf{i}}}\left(\sum_{\mathbf{l}<\mathbf{i}}\tilde{e}_{\mathbf{l}}\tilde{e}_{\mathbf{j}}\right)+\mathbb{E}_{X_{\mathbf{i}}}((\sum_{\mathbf{j}< \mathbf{i}} \tilde{e}_{\mathbf{j}})^2)
\\&\stackrel{(f)}{=}\mathbb{E}_{X_{\mathbf{i}}}\left(\mathbb{E}_{X^n|X_{\mathbf{i}}}^2(\tilde{e}|X_{\mathbf{i}})\right)-2\sum_{\mathbf{j}<\mathbf{i}}\sum_{\mathbf{l}<\mathbf{i}}\mathbbm{1}(\mathbf{j}=\mathbf{l})\mathbb{E}_{X_{\mathbf{i}}}\left(\tilde{e}_{\mathbf{l}}\tilde{e}_{\mathbf{j}}\right)+\mathbb{E}_{X_{\mathbf{i}}}((\sum_{\mathbf{j}< \mathbf{i}} \tilde{e}_{\mathbf{j}})^2)
\\&=\mathbb{E}_{X_{\mathbf{i}}}\left(\mathbb{E}_{X^n|X_{\mathbf{i}}}^2(\tilde{e}|X_{\mathbf{i}})\right)-2\sum_{\mathbf{j}<\mathbf{i}}\mathbb{E}_{X_{\mathbf{j}}}(\tilde{e}^2_{\mathbf{j}})+\mathbb{E}_{X_{\mathbf{i}}}((\sum_{\mathbf{j}< \mathbf{i}} \tilde{e}_{\mathbf{j}})^2)
 \\&=\mathbb{E}_{X_{\mathbf{i}}}(\mathbb{E}_{X^n|X_{\mathbf{i}}}^2(\tilde{e}|X_{\mathbf{i}}))-2\sum_{\mathbf{j}<\mathbf{i}}\mathbb{E}_{X_{\mathbf{j}}}(\tilde{e}^2_{\mathbf{j}})+\sum_{\mathbf{j}< \mathbf{i}} \sum_{\mathbf{k}<\mathbf{i}}\mathbb{E}_{X_{\mathbf{i}}}(\tilde{e}_{\mathbf{j}}\tilde{e}_{\mathbf{k}})
 \\&
\stackrel{(g)}{=}\mathbb{E}_{X_{\mathbf{i}}}(\mathbb{E}_{X^n|X_{\mathbf{i}}}^2(\tilde{e}|X_{\mathbf{i}}))-2\sum_{\mathbf{j}<\mathbf{i}}\mathbb{E}_{X_{\mathbf{j}}}(\tilde{e}^2_{\mathbf{j}})+\sum_{\mathbf{j}< \mathbf{i}} \sum_{\mathbf{k}<\mathbf{i}}\mathbbm{1}(\mathbf{j}=\mathbf{k})\mathbb{E}_{X_{\mathbf{i}}}(\tilde{e}^2_{\mathbf{j}})
 \\&=\mathbb{E}_{X_{\mathbf{i}}}(\mathbb{E}_{X^n|X_{\mathbf{i}}}^2(\tilde{e}|X_{\mathbf{i}}))-\sum_{\mathbf{j}< \mathbf{i}}\mathbf{P}_{\mathbf{j}},
\end{align*}
where (a) follows from 1) in Proposition \ref{prop:belong2}, b) follows from the decomposition in Equation \eqref{eq:Hil_Dec2}, (c) uses linearity of expectation, (d) uses 4) in Proposition \ref{prop:belong2}, (e) holds from 2) in  \ref{prop:belong2}, and in (f) and (g) we have used 1) in Proposition \ref{prop:belong2}.
\end{IEEEproof}
\begin{Corollary}
 For an arbitrary $e:\{0,1\}^n\to \{0,1\}$ with corresponding real function $\tilde{e}$, and decomposition $\tilde{e}=\sum_{\mathbf{j}}\tilde{e}_{\mathbf{j}}$. Let the variance of $\tilde{e}$ be denoted by $\mathbf{P}$. Then, $\mathbf{P}=\sum_{\mathbf{j}}\mathbf{P}_{\mathbf{j}}$. 
 \label{Cor:power}
\end{Corollary}
The corollary is a special case of Lemma \ref{Lem:power}, where we have taken $\mathbf{i}$ to be the all ones vector.
The following provides a definition of the dependency spectrum of a Boolean function:
\begin{Definition}[Dependency Spectrum]
 For a Boolean function $e$, the vector of variances $(P_{\mathbf{i}})_{\mathbf{i}\in \{0,1\}^n}$ is called the dependency spectrum of $e$.
\end{Definition}
In the next section, we will use the dependency spectrum to upper-bound the maximum correlation between the outputs of two arbitrary Boolean functions.

\section{Correlation Preservation in Arbitrary Functions}\label{sec:corr}
We proceed with presenting the main result of this paper. Let $(X,Y)$ be a pair of DMS's. Consider two arbitrary Boolean functions $e:\mathcal{X}^n\to \{0,1\}$ and $f:\mathcal{Y}^n\to \{0,1\}$. Let  $ q\triangleq P(e=1)$, $r\triangleq P(f=1)$. Let $\tilde{e}=\sum_{\mathbf{i}}e_{\mathbf{i}}$, and $\tilde{f}=\sum_{\mathbf{i}}f_{\mathbf{i}}$ give the decomposition of these functions as defined in the previous section.
The following theorem provides an  upper-bound on the probability of equality of $e(X^n)$ and $f(Y^n)$. 

\begin{Theorem}
 Let $\epsilon\triangleq P(X\neq Y)$, the following bound holds:
 \begin{align*}
&2\!\!\sqrt{\sum_{\mathbf{i}}\mathbf{P}_{\mathbf{i}}}\sqrt{\sum_{\mathbf{i}}\mathbf{Q}_{\mathbf{i}}}-2\!\!\sum_{\mathbf{i}}C_\mathbf{i}\mathbf{P}_{\mathbf{i}}^{\frac{1}{2}}\mathbf{Q}_{\mathbf{i}}^{\frac{1}{2}} 
\leq  \!\!P(e(X^n)\neq f(Y^n))
\\&\leq 1- 2\sqrt{\sum_{\mathbf{i}}\mathbf{P}_{\mathbf{i}}}\sqrt{\sum_{\mathbf{i}}\mathbf{Q}_{\mathbf{i}}}+2\sum_{\mathbf{i}}C_\mathbf{i}\mathbf{P}_{\mathbf{i}}^{\frac{1}{2}}\mathbf{Q}_{\mathbf{i}}^{\frac{1}{2}} 
,
\end{align*}
 where $C_{\mathbf{i}}\triangleq  (1-2\epsilon)^{N_\mathbf{i}}$, $\mathbf{P}_{\mathbf{i}}$ is the variance of $\tilde{e}_{\mathbf{i}}$, and ${\tilde{e}}$ is the real function corresponding to ${e}$, and $\mathbf{Q}_{\mathbf{i}}$ is the variance of $\tilde{f}_{\mathbf{i}}$, and finally, $N_{\mathbf{i}}\triangleq w_H(\mathbf{i})$.
 \label{th:sec3}
\end{Theorem}
\begin{IEEEproof}
 Please refer to the appendix.
\end{IEEEproof}

\begin{Remark}
 $C_{\mathbf{i}}$ is decreasing with $N_{\mathbf{i}}$. So, in order to increase $ P(e(X^n)\neq f(Y^n))$, most of the variance $\mathbf{P}_{\mathbf{i}}$ should be distributed on $\tilde{e}_\mathbf{i}$ which have lower $N_{\mathbf{i}}$ (i.e. operate on smaller blocks). Particularly, the lower bound is minimized by setting
 \begin{align*}
 \mathbf{P}_{\mathbf{i}}=
\begin{cases}
 1  \qquad  & \mathbf{i}=\mathbf{i}_1,\\
0 \qquad & otherwise.
\end{cases}
\end{align*}
This recovers the result in \cite{ComInf2}.
 \end{Remark}

We derived a relation between the dependency spectrum of a Boolean function and its correlation preserving properties. This can be used in a variety of disciplines. For example, in communication problems, cooperation among different nodes in a network requires correlated outputs which can be linked to the dependency spectrum through the results derived here. On the other hand, there are restrictions on the dependency spectrum based on the rate-distortion requirements (better performance requires larger effective lengths). We investigate this in \cite{arxiv1}, and show that the large blocklength single-letter coding strategies used in networks are sub-optimal in various problems.
\section{Conclusion}\label{sec:con}
We derived a new bound on the maximum correlation between Boolean functions operating on pairs of sequences of random variable. The bound was presented as a function of the dependency spectrum of the functions. We developed a new mathematical apparatus for analyzing Boolean functions, provided formulas for decomposing the Boolean function into additive components, and for calculating the dependency spectrum of these functions. The new bound has wide ranging applications in security, control and information theory.
\appendix

\subsection{Proof of Lemma \ref{Lem:unique}}
\begin{proof}
 
The uniqueness of such a decomposition follows from the isomorphy relation stated in equation \eqref{eq:Hil_Dec2}. We prove that the $\tilde{e}_{\mathbf{i}}$ given in the lemma are indeed the decomposition into the components of the direct sum. Equivalently, we show that $1)$ $\tilde{e}=\sum_{\mathbf{i}}\tilde{e}_{\mathbf{i}}$, and $2)$ $\tilde{e}_{\mathbf{i}}\in \mathcal{G}_{i_1}\otimes \mathcal{G}_{i_2}\otimes\dotsb \otimes \mathcal{G}_{i_n}, \forall \mathbf{i}\in \{0,1\}^n$.

First we check the equality $\tilde{e}=\sum_{\mathbf{i}}\tilde{e}_{\mathbf{i}}$. Let $\mathbf{t}$ denote the n-length vector whose elements are all ones. We have:
\begin{align*}
\tilde{e}_{\mathbf{t}}=\mathbb{E}_{X^n|X_{\mathbf{t}}}(\tilde{e}|X_{\mathbf{t}})-\sum_{\mathbf{i}<\mathbf{t}}\tilde{e}_{\mathbf{i}}
\stackrel{(a)}{\Rightarrow} \tilde{e}_{\mathbf{t}}+\sum_{\mathbf{i}<\mathbf{t}}\tilde{e}_{\mathbf{i}}=\tilde{e}\stackrel{(b)}{\Rightarrow} \tilde{e}=\sum_{\mathbf{i}\in \{0,1\}^n}\tilde{e}_{\mathbf{i}},
\end{align*}
where in (a) we have used 1) $X_{\mathbf{t}}=X^n$ and 2) for any function $\tilde{f}$ of $X^n$, $\mathbb{E}_{X^n|X^n}(\tilde{f}|X^n)=\tilde{f}$, and (b) holds since $\mathbf{i}<\mathbf{t}\Leftrightarrow \mathbf{i}\neq \mathbf{t}$.
. It remains to show that $\tilde{e}_{\mathbf{i}}\in \mathcal{G}_{i_1}\otimes \mathcal{G}_{i_2}\otimes\dotsb \otimes \mathcal{G}_{i_n}, \forall \mathbf{i}\in \{0,1\}^n$. The next proposition provides a means to verify this property.
\begin{Proposition}
 Fix $\mathbf{i}\in \{0,1\}^n$, define $\mathcal{A}_0\triangleq\{s|i_s=0\}$, and $\mathcal{A}_1\triangleq\{s|i_s=1\}$. $\tilde{f}$ is an element of $\mathcal{G}_{i_1}\otimes \mathcal{G}_{i_2}\otimes\dotsb \otimes \mathcal{G}_{i_n}$ if and only if $(1)$ it is constant in all $X_s$, $s\in \mathcal{A}_0$, and $(2)$ it has $0$ mean on all $X_s$, when $s\in \mathcal{A}_1$. 
\label{prop:belong1}\end{Proposition}
\begin{proof}
By definition, any element of $\mathcal{G}_{i_1}\otimes \mathcal{G}_{i_2}\otimes\dotsb \otimes \mathcal{G}_{i_n}$ satisfies the conditions in the proposition. Conversely, we show that any function satisfying the conditions $(1)$ and $(2)$ is in the tensor product.  Let $\tilde{f}=\sum_{\mathbf{j}}\tilde{f}_{\mathbf{j}}, \tilde{f}_{\mathbf{j}}\in\mathcal{G}_{j_1}\otimes \mathcal{G}_{j_2}\otimes\dotsb \otimes \mathcal{G}_{j_n}$. Assume $i_k=1$ for some $k\in [1,n]$. Then: 
 \begin{align*}
 &0\stackrel{(2)}{=}\mathbb{E}_{X^n|X_{\sim i_k}}(\sum_{\mathbf{j}}\tilde{f}_{\mathbf{j}}|X_{\sim i_k})\stackrel{(a)}{=}\sum_{\mathbf{j}}\mathbb{E}_{X^n|X_{\sim i_k}}(\tilde{f}_{\mathbf{j}}|X_{\sim i_k})
 \\&\stackrel{(1)}{=} 
 \sum_{\mathbf{j}: j_k=0}\mathbb{E}_{X^n|X_{\sim i_k}}(\tilde{f}_{\mathbf{i}}|X_{\sim i_k})\stackrel{(2)}{=}\sum_{\mathbf{j}: j_k=0}\tilde{f}_{\mathbf{j}}, 
 \end{align*}
 where we have used linearity of expectation in (a), and the last two equalities use the fact that $\tilde{f}_{\mathbf{j}}\in \mathcal{G}_{j_1}\otimes \mathcal{G}_{j_2}\otimes\dotsb \otimes \mathcal{G}_{j_n}$ which means it satisfies properties $(1)$ and $(2)$. So far we have shown that $\tilde{f}=\sum_{\mathbf{j}\geq \mathbf{i}}\tilde{f}_{\mathbf{j}}$. Now assume $i_{k'}=0$. Then:
 \begin{align*}
&\sum_{\mathbf{j}\geq\mathbf{i}}\tilde{f}_{\mathbf{j}}=\tilde{f}\stackrel{(1)}{=}\mathbb{E}_{X^n|X_{\sim i_{k'}}}(\sum_{\mathbf{j}\geq\mathbf{i}}\tilde{f}_{\mathbf{j}}|X_{\sim i_{k'}})=
 \sum_{\mathbf{j}\geq\mathbf{i}}\mathbb{E}_{X^n|X_{\sim i_{k'}}}(\tilde{f}_{\mathbf{j}}|X_{\sim i_{k'}})
 \\&=  \sum_{\mathbf{j}\geq\mathbf{i}: j_{k'}=0}\tilde{f}_{\mathbf{j}} \Rightarrow 
 \sum_{\mathbf{j}\geq\mathbf{i}:j_{k'}=1}\tilde{f}_{\mathbf{j}}=0.
 \end{align*}
So, $\tilde{f}=\sum_{\mathbf{i}\geq\mathbf{j}\geq\mathbf{i}} \tilde{f}_{\mathbf{j}}=\tilde{f}_{\mathbf{i}}$. By assumption we have $\tilde{f}_{\mathbf{i}}\in\mathcal{G}_{i_1}\otimes \mathcal{G}_{i_2}\otimes\dotsb \otimes \mathcal{G}_{i_n}$. 
\end{proof}

 Returning to the original problem, it is enough to show that $\tilde{e}_{\mathbf{i}}$'s satisfy the conditions in Proposition \ref{prop:belong1}.  We prove the stronger result presented in the next proposition.
\begin{Proposition}
\label{prop:belong2}
The following hold:
\nonumber\\1) $\mathbb{E}_{X^n}(\tilde{e}_{\mathbf{i}})$=0.\\
 2) $\forall \mathbf{i}\leq \mathbf{k}$, we have $\mathbb{E}_{X^n|X_{\mathbf{j}}}(\tilde{e}_{\mathbf{i}}|X_{\mathbf{k}})=\tilde{e}_{\mathbf{i}}$.\\
 3) $\mathbb{E}_{X^n}(\tilde{e}_{\mathbf{i}}\tilde{e}_{\mathbf{k}})=0$, for $\mathbf{i}\neq \mathbf{k}$.\\
 4) $\forall \mathbf{k}\leq \mathbf{i}: \mathbb{E}_{X^n|X_{\mathbf{k}}}(\tilde{e}_{\mathbf{i}}|X_{\mathbf{k}})=0.$
 \end{Proposition}
\begin{proof}
 1) For two n-length binary vectors $\mathbf{i}$, and $\mathbf{j}$, we write $\mathbf{i}\leq \mathbf{j}$ if $i_k\leq j_k, \forall k\in[1,n]$. The set $\{0,1\}^n$ equipped with $\leq$ is a well-founded set (i.e. any subset of $\{0,1\}^n$ has at least one minimal element).  The following presents the principle of Noetherian induction on well-founded sets:
 \begin{Proposition}[Principle of Noetherian Induction]
 Let $(A,\preccurlyeq)$ be a well-founded set. To prove the property $P(x)$ is true for all elements $x$ in $A$, it is sufficient to prove the following
 \\1) \textbf{Induction Basis:} $P(x)$ is true for all minimal elements in $A$. 
 \\2) \textbf{Induction Step:} For any non-minimal element $x$ in $A$, if $P(y)$ is true for all minimal $y$ such that $y\prec x$, then it is true for $x$.
 \end{Proposition}
 We will use Noetherian induction to prove the result.
  Let $\mathbf{i}_j, j\in [1,n]$ be the $j$th element of the standard basis. Then $\tilde{e}_{\mathbf{i}_j}= \mathbb{E}_{X^n|X_{ j}}(\tilde{e}|X_{ j})$. By the smoothing property of expectation, $\mathbb{E}_{X^n}(\tilde{e}_{\mathbf{i}_j})=\mathbb{E}_{X^n}(\tilde{e})=0$. Assume that $\forall \mathbf{j}<\mathbf{i}$,  $\mathbb{E}_{X^n}(\tilde{e}_{\mathbf{j}})=0$. Then,
\begin{align*}
 \mathbb{E}_{X^n}(\tilde{e}_{\mathbf{i}})
&=\mathbb{E}_{X^n}\left(\mathbb{E}_{X^n|X_{\mathbf{i}}}(\tilde{e}|X_{\mathbf{i}})-\sum_{\mathbf{j}< \mathbf{i}} \tilde{e}_{\mathbf{j}}\right)
\\&= 
\mathbb{E}_{X^n}(\tilde{e})- \sum_{\mathbf{j}< \mathbf{i}}  \mathbb{E}_{X^n}(\tilde{e}_{\mathbf{j}})
=0- \sum_{\mathbf{j}< \mathbf{i}}  0 =0.
\end{align*}
\\2) This statement is also proved by induction. $\mathbb{E}_{X^n|X_{\mathbf{i}}}(\tilde{e}|X_{\mathbf{i}})$ is a function of $X_{\mathbf{i}}$, so by induction $\tilde{e}_{\mathbf{i}}=\mathbb{E}_{X^n|X_{\mathbf{i}}}(\tilde{e}|X_{\mathbf{i}})-\sum_{\mathbf{j}<\mathbf{i}}\tilde{e}_{\mathbf{k}}$ is also a function of $X_{\mathbf{i}}$. 
\\3) Let $\mathbf{i}_k, k\in [1,n]$ be defined as the $k$th element of the standard basis, and take $j,j'\in [1,n], j\neq j'$. We have:
\begin{align*}
 &\mathbb{E}_{X^n}(\tilde{e}_{\mathbf{i}_j}\tilde{e}_{\mathbf{i}_{j'}})=\mathbb{E}_{X^n}(\mathbb{E}_{X^n|X_j}(\tilde{e}|X_j)\mathbb{E}_{X^n|X_{j'}}(\tilde{e}|X_{j'}))
\\& \stackrel{(a)}{=} \mathbb{E}_{X^n}(\mathbb{E}_{X^n|X_j}(\tilde{e}|X_j)) \mathbb{E}_{X^n}(\mathbb{E}_{X^n|X_{j'}}(\tilde{e}|X_{j'}))
 \stackrel{(b)}{=} \mathbb{E}^2_{X^n}(\tilde{e})=0,
\end{align*}
where we have used the memoryless property of the source in (a) and (b) results from the smoothing property of expectation. We extend the argument by Noetherian induction. Fix $\mathbf{i}, \mathbf{k}$. Assume that $\mathbb{E}_{X^n}(\tilde{e}_{\mathbf{j}}\tilde{e}_{\mathbf{j}'})=\mathbbm{1}(\mathbf{j}=\mathbf{j}')\mathbb{E}_{X^n}(\tilde{e}^2_{\mathbf{j}}), \forall  \mathbf{j}< \mathbf{i}, \mathbf{j}'\leq \mathbf{k}$, and $\forall  \mathbf{j}\leq \mathbf{i}, \mathbf{j}'< \mathbf{k}$. 
\begin{align*}
  &\mathbb{E}_{X^n}(\tilde{e}_{\mathbf{i}}\tilde{e}_{\mathbf{k}})
  =\mathbb{E}_{X^n}\left(\left(\mathbb{E}_{X^n|X_{\mathbf{i}}}(\tilde{e}|X_{\mathbf{i}})-\sum_{\mathbf{j}< \mathbf{i}} \tilde{e}_{\mathbf{j}}\right)\left(\mathbb{E}_{X^n|X_{\mathbf{k}}}(\tilde{e}|X_{\mathbf{k}})-\sum_{\mathbf{j}'< \mathbf{k}} \tilde{e}_{\mathbf{j}'}\right)\right)\\
  &=\mathbb{E}_{X_n}\left(\mathbb{E}_{X^n|X_{\mathbf{i}}}(\tilde{e}|X_{\mathbf{i}})\mathbb{E}_{X^n|X_{\mathbf{k}}}(\tilde{e}|X_{\mathbf{k}})\right)
  -\sum_{\mathbf{j}< \mathbf{i}} \mathbb{E}_{X^n}\left(\tilde{e}_{\mathbf{j}}\mathbb{E}_{X^n|X_{\mathbf{k}}}(\tilde{e}|X_{\mathbf{k}})\right)
  \\&- \sum_{\mathbf{j}'< \mathbf{k}} \mathbb{E}_{X^n}\left(\tilde{e}_{\mathbf{j}'}\mathbb{E}_{X^n|X_{\mathbf{i}}}(\tilde{e}|X_{\mathbf{i}})\right)
  + \sum_{\mathbf{j}< \mathbf{i}, \mathbf{j}'<\mathbf{k}}\mathbb{E}_{X^n}(\tilde{e}_{\mathbf{j}}\tilde{e}_{\mathbf{j}'}).
\end{align*}
The second and third terms in the above expression can be simplified as follows. First, note that:
\begin{align}
&\tilde{e}_{\mathbf{i}}=\mathbb{E}_{X^n|X_{\mathbf{i}}}(\tilde{e}|X_{\mathbf{i}})-\sum_{\mathbf{j}< \mathbf{i}} \tilde{e}_{\mathbf{j}}\Rightarrow \sum_{\mathbf{j}\leq \mathbf{i}} \tilde{e}_{\mathbf{j}}=\mathbb{E}_{X^n|X_{\mathbf{i}}}(\tilde{e}|X_{\mathbf{i}}).
\label{eq:sum_int}
\end{align}
Our goal is to simplify $\mathbb{E}_{X^n}(\tilde{e}_{\mathbf{j}}\mathbb{E}_{X^n|X_{\mathbf{j}'}}(\tilde{e}|X_{\mathbf{j}'}))$. We proceed by considering two different cases:
\\\textbf{Case 1:} $\mathbf{i}\nleq \mathbf{k}$ and $\mathbf{k}\nleq \mathbf{i}$: 

Let $\mathbf{j}<\mathbf{i}$:
\begin{align*}
&\mathbb{E}_{X^n}(\tilde{e}_{\mathbf{j}}\mathbb{E}_{X^n|X_{\mathbf{k}}}(\tilde{e}|X_{\mathbf{k}}))
\stackrel{\eqref{eq:sum_int}}{=} \mathbb{E}_{X^n}(\tilde{e}_{\mathbf{j}}\sum_{\mathbf{l}\leq \mathbf{k}} \tilde{e}_{\mathbf{j}}))
\\&=\sum_{\mathbf{l}\leq \mathbf{k}} \mathbb{E}_{X^n}(\tilde{e}_{\mathbf{j}}\tilde{e}_{\mathbf{l}})=\sum_{\mathbf{l}\leq \mathbf{k}}\mathbbm{1}({\mathbf{j}=\mathbf{l}})  \mathbb{E}_{X^n}(\tilde{e}_{\mathbf{j}}^2)=\mathbbm{1}({\mathbf{j}\leq\mathbf{k}})  \mathbb{E}_{X^n}(\tilde{e}_{\mathbf{j}}^2).
\end{align*}
By the same arguments, for $\mathbf{j}'\leq\mathbf{k}$:
\begin{align*}
  \mathbb{E}_{X^n}\left(\tilde{e}_{\mathbf{j}'}\mathbb{E}_{X^n|X_{\mathbf{i}}}(\tilde{e}|X_{\mathbf{i}})\right)=
  \mathbbm{1}({\mathbf{j}'\leq\mathbf{i}})  \mathbb{E}_{X^n}(\tilde{e}_{\mathbf{j}'}^2).
\end{align*}

Replacing the terms in the original equality we get:
\begin{align*}
   \mathbb{E}_{X^n}(\tilde{e}_{\mathbf{i}}\tilde{e}_{\mathbf{k}})
   &= \mathbb{E}_{X^n}\left(\mathbb{E}_{X^n|X_{\mathbf{i}}}(\tilde{e}|X_{\mathbf{i}})\mathbb{E}_{X^n|X_{\mathbf{k}}}(\tilde{e}|X_{\mathbf{k}})\right)
   -\sum_{\mathbf{j}< \mathbf{i}}\mathbbm{1}({\mathbf{j}\leq\mathbf{k}})  \mathbb{E}_{X^n}(\tilde{e}_{\mathbf{j}}^2)
   \\&- \sum_{\mathbf{j}'\leq \mathbf{k}} \mathbbm{1}({\mathbf{j}'\leq\mathbf{i}})  \mathbb{E}_{X^n}(\tilde{e}_{\mathbf{j}'}^2)+ \sum_{\mathbf{j}< \mathbf{i}, \mathbf{j}'<\mathbf{k}}\mathbbm{1}(\mathbf{j}=\mathbf{j}')\mathbb{E}_{X^n}(\tilde{e}_{\mathbf{j}}^2)\\
   &= \mathbb{E}_{X^n}\left(\mathbb{E}_{X^n|X_{\mathbf{i}}}(\tilde{e}|X_{\mathbf{i}})\mathbb{E}_{X^n|X_{\mathbf{k}}}(\tilde{e}|X_{\mathbf{k}})\right)-\sum_{\mathbf{j}\leq \mathbf{i}\wedge \mathbf{k}}\mathbb{E}_{X^n}(\tilde{e}_{\mathbf{j}}^2)\\
   &\stackrel{(a)}{=}\mathbb{E}_{X^n}(\mathbb{E}_{X^n|X_{ \mathbf{i}\wedge \mathbf{k}}}^2(\tilde{e}(X^n)|X_{ \mathbf{i}\wedge \mathbf{k}}))-\sum_{\mathbf{j}\leq \mathbf{i}\wedge \mathbf{k}}\mathbb{E}_{X^n}(\tilde{e}_{\mathbf{j}}^2)
   \\&\stackrel{(b)}{=}\mathbb{E}_{X^n}(\mathbb{E}_{X^n|X_{ \mathbf{i}\wedge \mathbf{k}}}^2(\tilde{e}(X^n)|X_{ \mathbf{i}\wedge \mathbf{k}}))-\mathbb{E}_{X^n}\left((\sum_{\mathbf{j}\leq \mathbf{i}\wedge \mathbf{k}}\tilde{e}_{\mathbf{j}})^2\right)\stackrel{\eqref{eq:sum_int}}{=}0
\end{align*}
Where in (b) we have used that $\tilde{e}_{\mathbf{i}}$'s are uncorrelated, and (a) is proved below:
\begin{align*}
&\mathbb{E}_{X^n}\left(\mathbb{E}_{X^n|X_{\mathbf{i}}}(\tilde{e}|X_{\mathbf{i}})\mathbb{E}_{X^n|X_{\mathbf{k}}}(\tilde{e}|X_{\mathbf{k}})\right)
\\&= \sum_{x_{ \mathbf{i}\wedge \mathbf{k}}}P(x_{\mathbf{i}\wedge \mathbf{k}})
\\&\left( \left(\sum_{x_{{|\mathbf{i}}- \mathbf{k}|^+}}P(x_{|{\mathbf{i}}- \mathbf{k}|^+})
\mathbb{E}_{X^n|X_{\mathbf{i}}}(\tilde{e}|X_{\mathbf{i}})\right)\left(\sum_{x_{{|\mathbf{k}}- \mathbf{i}|^+}}P(x_{|{\mathbf{k}}- \mathbf{i}|^+})\mathbb{E}_{X^n|X_{\mathbf{k}}}(\tilde{e}|X_{\mathbf{k}}\right)\right)\\
&=\sum_{x_{ \mathbf{i}\wedge \mathbf{k}}}P(x_{\mathbf{i}\wedge \mathbf{k}}) \mathbb{E}_{X^n|X_{\mathbf{i}\wedge\mathbf{k}}}^2(\tilde{e}|x_{\mathbf{i}\wedge\mathbf{k}})\\
&=\mathbb{E}_{X^n}(\mathbb{E}_{X^n|X_{ \mathbf{i}\wedge \mathbf{k}}}^2(\tilde{e}(X^n)|X_{ \mathbf{i}\wedge \mathbf{k}})).
\end{align*}
\\\textbf{Case 2:} Assume $\mathbf{i}\leq \mathbf{k}$:
\begin{align*}
   & \mathbb{E}_{X^n}(\tilde{e}_{\mathbf{i}}\tilde{e}_{\mathbf{k}})=  \mathbb{E}_{X^n}\left(\mathbb{E}_{X^n|X_{\mathbf{i}}}(\tilde{e}|X_{\mathbf{i}})\mathbb{E}_{X^n|X_{\mathbf{k}}}(\tilde{e}|X_{\mathbf{k}})\right)-\sum_{\mathbf{j}< \mathbf{i}}\mathbbm{1}({\mathbf{j}\leq\mathbf{j}'})  \mathbb{E}_{X^n}(\tilde{e}_{\mathbf{j}}^2)
    \\&- \sum_{\mathbf{j}'\leq \mathbf{k}} \mathbbm{1}({\mathbf{j}'\leq\mathbf{j}}) 
     \mathbb{E}_{X^n}(\tilde{e}_{\mathbf{j}'}^2) + \sum_{\mathbf{j}< \mathbf{i}, \mathbf{j}'<\mathbf{k}}\mathbbm{1}(\mathbf{j}=\mathbf{j}')
     \mathbb{E}_{X^n}(\tilde{e}_{\mathbf{j}}^2)\\
    &=\mathbb{E}_{X^n}(\mathbb{E}_{X^n|X_{\mathbf{i}}}^2(\tilde{e}|X_{\mathbf{i}}))-\sum_{\mathbf{j}<\mathbf{i}}\mathbb{E}_{X^n}(\tilde{e}^2_{\mathbf{j}})-\sum_{\mathbf{j}'\leq\mathbf{i}}\mathbb{E}_{X^n}(\tilde{e}^2_{\mathbf{j}'})+\sum_{\mathbf{j}\leq\mathbf{i}}\mathbb{E}_{X^n}(\tilde{e}^2_{\mathbf{j}})\\
    &=0.
\end{align*}
\\\textbf{Case 3:} When $\mathbf{k}\leq \mathbf{i}$ the proof is similar to case 2.
\\4) Clearly when $|\mathbf{i}|=1$, the claim holds. Assume it is true for all $\mathbf{j}$ such that $|\mathbf{j}|<\mathbf{i}$. 
Take $\mathbf{i}\in \{0,1\}^n$ and $t\in [1,n], i_t=1$ arbitrarily. We first prove the claim for $\mathbf{k}=\mathbf{i}-\mathbf{i}_t$:
\begin{align*}
 &\mathbb{E}_{X^n|X_{\mathbf{k}}}(\tilde{e}_{\mathbf{i}}|X_{\mathbf{k}})
 = \mathbb{E}_{X^n|X_{\mathbf{k}}}\left(\left(\mathbb{E}_{X^n|X_{\mathbf{i}}}(\tilde{e})-\sum_{\mathbf{j}<\mathbf{i}}\tilde{e}_{\mathbf{j}}\right)|X_{\mathbf{k}}\right)
 \\&= \mathbb{E}_{X^n|X_{\mathbf{k}}}\left(\mathbb{E}_{X^n|X_{\mathbf{i}}}(\tilde{e}|X_{\mathbf{i}})|X_{\mathbf{k}}\right)-\sum_{\mathbf{j}<\mathbf{i}}\mathbb{E}_{X^n|X_{\mathbf{k}}}(\tilde{e}_{\mathbf{j}}|X_{\mathbf{k}})
 \\&\stackrel{(a)}{=}\mathbb{E}_{X^n|X_{\mathbf{k}}}(\tilde{e}|X_{\mathbf{k}})-\sum_{\mathbf{j}<\mathbf{i}}\mathbb{E}_{X^n|X_{\mathbf{k}}}(\tilde{e}_{\mathbf{j}}|X_{\mathbf{k}})
 \stackrel{(5)}{=} \sum_{\mathbf{j}\leq \mathbf{i}-\mathbf{i}_t}\tilde{e}_{\mathbf{j}} -\sum_{\mathbf{j}<\mathbf{i}}\mathbb{E}_{X^n|X_{\mathbf{k}}}(\tilde{e}_{\mathbf{j}}|X_{\mathbf{k}})
 \\&\stackrel{(b)}{=} \sum_{\mathbf{j}\leq \mathbf{i}-\mathbf{i}_t}\mathbb{E}_{X^n|X_{\mathbf{k}}}(\tilde{e}_{\mathbf{j}}|X_{\mathbf{k}}) -\sum_{\mathbf{j}<\mathbf{i}}\mathbb{E}_{X^n|X_{\mathbf{k}}}(\tilde{e}_{\mathbf{j}}|X_{\mathbf{k}})=\sum_{s\neq t}\mathbb{E}_{X^n|X_{\mathbf{k}}}(\tilde{e}_{\mathbf{i}-\mathbf{i}_s}|X_{\mathbf{k}})
 \\&\stackrel{(c)
} {=}\sum_{s\neq t}\mathbb{E}_{X^n|X_{\mathbf{k}-\mathbf{i}_s}}(\tilde{e}_{\mathbf{i}-\mathbf{i}_s}|X_{\mathbf{k}-\mathbf{i}_s})\stackrel{(d)}=0.
\end{align*}
Where in (a) we have used $\mathbf{i}>\mathbf{k}$, also (b) follows from $\mathbf{j}<\mathbf{k}$, (c) uses $\mathbf{k}\wedge(\mathbf{i}-\mathbf{i}_s)= \mathbf{k}-\mathbf{i}_s$, and finally, (d) uses the induction assumption. Now we extend the result to general $\mathbf{k}<\mathbf{i}$. Fix  $\mathbf{k}$. Assume the claim is true for all $\mathbf{j}$ such that $\mathbf{k}<\mathbf{j}<\mathbf{i}$ (i.e $\forall \mathbf{k}<\mathbf{j}<\mathbf{i}, \mathbb{E}_{X^n|X_{\mathbf{k}}}(\tilde{e}_{X_{\mathbf{j}}|X_{\mathbf{k}}})=0$). We have:
\begin{align*}
 &\mathbb{E}_{X^n|X_{\mathbf{k}}}(\tilde{e}_{\mathbf{i}}|X_{\mathbf{k}})= \mathbb{E}_{X^n|X_{\mathbf{k}}}\left(\mathbb{E}_{X^n|X_{\mathbf{i}}}(\tilde{e}|X_{\mathbf{i}})-\sum_{\mathbf{j}<\mathbf{i}}\tilde{e}_{\mathbf{j}}|X_{\mathbf{k}}\right)
\\& = \mathbb{E}_{X^n|X_{\mathbf{k}}}\left(\mathbb{E}_{X^n|X_{\mathbf{i}}}(\tilde{e}|X_{\mathbf{i}})|X_{\mathbf{k}}\right)-\sum_{\mathbf{j}\leq \mathbf{k}}\mathbb{E}_{X^n|X_{\mathbf{k}}}(\tilde{e}_{\mathbf{j}}|X_{\mathbf{k}})
\\& =\mathbb{E}_{X^n|X_{\mathbf{k}}}(\tilde{e}|X_{\mathbf{k}})-\sum_{\mathbf{j}\leq \mathbf{k}}\tilde{e}_{\mathbf{j}}\stackrel{\eqref{eq:sum_int}}{=}0.
\end{align*}

\end{proof}
\begin{Remark}
 The second condition above is equivalent to condition (2) in Proposition \ref{prop:belong1}. The fourth condition is equivalent to (1) in Proposition \ref{prop:belong1}.
\end{Remark}
Using propositions \ref{prop:belong1} and \ref{prop:belong2}, we conclude that $\tilde{e}_{\mathbf{i}}\in \mathcal{G}_{i_1}\otimes \mathcal{G}_{i_2}\otimes\dotsb \otimes \mathcal{G}_{i_n}, \forall \mathbf{i}\in \{0,1\}^n$. This completes the proof of Lemma \ref{Lem:unique}.
\end{proof}

\subsection{Proof of Theorem 1}
\begin{IEEEproof}
The proof involves three main steps. First, we bound the Pearson correlation between the real-valued functions $\tilde{e}$, and $\tilde{f}$. In the second step, we relate the correlation to the probability that the two functions are equal and derive the lower bound. Finally, in the third step we use the lower bound proved in the first two steps to derive the upper bound. 
\\\textbf{Step 1:} From Remark \ref{Rem:exp_0}, the expectation of both functions is 0. So, the Pearson correlation is given by $\frac{\mathbb{E}_{X^n,Y^n}(\tilde{e}\tilde{f})}{\left(rq(1-q)(1-r)\right)^{\frac{1}{2}}}$. Our goal is to bound this value. We have:
\begin{align}
\nonumber &\mathbb{E}_{X^n,Y^n}(\tilde{e}\tilde{f})
 \stackrel{(a)}=\mathbb{E}_{X^n,Y^n}\left((\sum_{\mathbf{i}\in\{0,1\}^n}\tilde{e}_{\mathbf{i}})(\sum_{\mathbf{k}\in\{0,1\}^n}\tilde{f}_{\mathbf{k}})\right)
 \\&\stackrel{(b)}{=}\sum_{\mathbf{i}\in\{0,1\}^n}\sum_{\mathbf{k}\in\{0,1\}^n}\mathbb{E}_{X^n,Y^n}( \tilde{e}_{\mathbf{i}}\tilde{f}_{\mathbf{k}}).
 \label{Eq:init1}
\end{align}
In (a) we have used Remark \ref{Rem:Dec}, and in (b) we use linearity of expectation.
 Using the fact that $\tilde{e}_{\mathbf{i}}\in \mathcal{G}_{i_1}\otimes \mathcal{G}_{i_2}\otimes\dotsb \otimes \mathcal{G}_{i_n}$ and Lemma \ref{Lem:tensor_dec}, we have: 
\begin{align}
 \tilde{e}_{\mathbf{i}}=c_{\mathbf{i}}\prod_{t:i_t=1} \tilde{e}_{t}(X_{t}),  \tilde{f}_{\mathbf{k}}=d_{\mathbf{k}}\prod_{t:k_t=1} \tilde{f}_{t}(X_{t}).
 \label{eq:init2}
\end{align}
We replace $\tilde{e}_{\mathbf{i}}$ and $ \tilde{f}_{\mathbf{k}}$ in \eqref{Eq:init1}:  
\begin{align}
\nonumber& \mathbb{E}_{X^n,Y^n}( \tilde{e}_{\mathbf{i}}\tilde{f}_{\mathbf{k}})
\stackrel{\eqref{eq:init2}}{=}\mathbb{E}_{X^n,Y^n}\left(\left(c_{\mathbf{i}} \prod_{t:i_t=1} \tilde{e}_t(X_{t})\right)\left(d_{\mathbf{k}} \prod_{s:k_s=1} \tilde{f}_s(Y_{s})\right)\right)
\\&\nonumber\stackrel{(a)}{=}c_{\mathbf{i}}d_{\mathbf{k}} \mathbb{E}_{X^n,Y^n}\left(\prod_{t:i_t=1, k_t=1}\tilde{e}_t(X_{t})\tilde{f}_t(Y_{t})\right)
\end{align}
\begin{align}&\mathbb{E}_{X^n}\left(\prod_{t:i_t=1, k_t=0} \tilde{e}_t(X_{t})\right)\mathbb{E}_{Y^n}\left( \prod_{t:i_t=0, k_t=1} \tilde{f}_t(Y_{t})\right)\nonumber\\
&\nonumber\stackrel{(b)}{=}\mathbbm{1}(\mathbf{i}=\mathbf{k})c_{\mathbf{i}}d_{\mathbf{k}}\prod_{t:i_t=1} \mathbb{E}_{X^n,Y^n}\left(\tilde{e}_t(X_{t}) \tilde{f}_t(Y_{t})\right)
\\&\stackrel{(c)}{\leq} \mathbbm{1}(\mathbf{i}=\mathbf{k})c_{\mathbf{i}}d_{\mathbf{k}}  (1-2\epsilon)^{N_{\mathbf{i}}} \prod_{t:i_t=1}\mathbb{E}_{X^n}^{\frac{1}{2}}\left(\tilde{e}_t^2(X_{t})\right)\mathbb{E}_{Y^n}^{\frac{1}{2}}\left( \tilde{f}^2_{t}(Y_t)\right)\nonumber\\
&
\stackrel{(d)}{=} \mathbbm{1}(\mathbf{i}=\mathbf{k}) (1-2\epsilon)^{N_{\mathbf{i}}}\mathbf{P}^{\frac{1}{2}}_{\mathbf{i}}\mathbf{Q}^{\frac{1}{2}}_{\mathbf{i}}
=  \mathbbm{1}(\mathbf{i}=\mathbf{k}) C_{\mathbf{i}}\mathbf{P}^{\frac{1}{2}}_{\mathbf{i}}\mathbf{Q}^{\frac{1}{2}}_{\mathbf{i}}.
\label{eq:intemid1}
\end{align}
 In (a) we have used the fact that in a pair of DMS's, $X_i$ and $Y_j$ are independent for $i\neq j$. (b) holds since from Proposition \ref{prop:belong2}, $\mathbb{E}(\tilde{e}_{i})=\mathbb{E}(\tilde{f}_i)=0, \forall i\in [1,n]$. We prove (c) in Lemma \ref{Lem:init} below. In (d) we have used proposition \ref{pr:partfun}.
\begin{Lemma}
\label{Lem:init}
 Let $g(X)$ and $h(Y)$ be two arbitrary zero-mean, real valued functions, then:
\begin{align*}
 \mathbb{E}_X(g(X)h(Y))\leq (1-2\epsilon)\mathbb{E}_X^{\frac{1}{2}}(g^2(X))\mathbb{E}_Y^{\frac{1}{2}}(h^2(Y)).
\end{align*}
 \end{Lemma}
\begin{proof}
Please refer to the \cite{arxiv1}. 
\end{proof}
 Using equations \eqref{Eq:init1} and \eqref{eq:intemid1} we get: 
\begin{align*}
 \mathbb{E}_X(\tilde{e}\tilde{f})\leq  \sum_{\mathbf{i}}C_{\mathbf{i}}\mathbf{P}^{\frac{1}{2}}_{\mathbf{i}}\mathbf{Q}^{\frac{1}{2}}_{\mathbf{i}}.
\end{align*}
\\\textbf{Step 2:} We use the results from step one to derive a bound on $P(e\neq f)$. Define $a\triangleq P(e(X^n)=1,f(Y^n)=1)$, $b\triangleq P(e(X^n)=0,f(Y^n)=1)$, $c\triangleq P(e(X^n)=1,f(Y^n)=0)$, and $d\triangleq P(e(X^n)=0, f(Y^n)=0)$, then
 \begin{align}
\label{eq:var}
& \mathbb{E}_{X^n,Y^N}(\tilde{e}(X^n)\tilde{f}(Y^n))\nonumber
\\&=a(1-q)(1-r)-bq(1-r)-c(1-q)r+dqr,
\end{align}
We write this equation in terms of $\sigma\triangleq P(f\neq g)$, q, and r using the following relations:
\begin{align*}
&1) a+c=q, \qquad 2) b+d=1-q,\qquad  
\\&3) a+b=r,\qquad  4) c+d=1-r, \qquad 5) b+c=\sigma.
\end{align*}
 Solving the above we get:
\begin{align}
&\nonumber a=\frac{q+r-\sigma}{2}, \qquad b=\frac{r+\sigma-q}{2},\qquad 
 \\&c=\frac{q-r+\sigma}{2},\qquad d=1-\frac{q+r+\sigma}{2}.
\label{eq:intermed}
\end{align}
 
We replace $a,b,c$, and $d$ in \eqref{eq:var} by their values in $\eqref{eq:intermed}$:
\begin{align*}
&\frac{\sigma}{2}\geq (\frac{q+r}{2})(1-q)(1-r)+(\frac{q-r}{2})q(1-r)
\\&+(\frac{r-q}{2})(1-q)r+qr(1-\frac{q+r}{2})-\sum_{\mathbf{i}}C_\mathbf{i}\mathbf{P}_{\mathbf{i}}^{\frac{1}{2}}\mathbf{Q}_{\mathbf{i}}^{\frac{1}{2}}\\
&\Rightarrow \sigma\geq q+r-2rq-2\sum_{\mathbf{i}}C_\mathbf{i}\mathbf{P}_{\mathbf{i}}^{\frac{1}{2}}\mathbf{Q}_{\mathbf{i}}^{\frac{1}{2}}
\\&\Rightarrow \sigma \geq (\sqrt{q(1-r)}-\sqrt{r(1-q)})^2
\\&+2\sqrt{q(1-q)r(1-r)}-2\sum_{\mathbf{i}}C_\mathbf{i}\mathbf{P}_{\mathbf{i}}^{\frac{1}{2}}\mathbf{Q}_{\mathbf{i}}^{\frac{1}{2}}\\
&\Rightarrow \sigma \geq 2\sqrt{q(1-q)r(1-r)}-2\sum_{\mathbf{i}}C_\mathbf{i}\mathbf{P}_{\mathbf{i}}^{\frac{1}{2}}\mathbf{Q}_{\mathbf{i}}^{\frac{1}{2}}
\end{align*}
 On the other hand $\mathbb{E}_X(\tilde{e}^2)=q(1-q)=\sum_{\mathbf{i}}\mathbf{P}_{\mathbf{i}}$, where the last equality follows from the fact that  $\tilde{e}_{\mathbf{i}}$'s are uncorrelated. This proves the lower bound. Next we use the lower bound to derive the upper bound.
\\\textbf{Step 3:} The upper-bound can be derived by considering the function $h(Y^n)$ to be the complement of $f(Y^n)$ (i.e. $h(Y^n)\triangleq 1\oplus_2 f(Y^n)$.) In this case $P({h}(Y^n)=1)=P(f(Y^n)=0)=1-r$. The corresponding real function for $h(Y^n)$ is:
\begin{align*}
 \tilde{h}(Y^n)&=
\begin{cases}
 r  \qquad  &h(Y^n)=1,\\
-(1-r) \qquad & h(Y^n)=0,
\end{cases}
\\&=
\begin{cases}
 r  \qquad  &f(Y^n)=0,\\
-(1-r) \qquad&  f(Y^n)=1,
\end{cases}
\Rightarrow \tilde{h}(Y^n)=-\tilde{f}(Y^n).
\end{align*} 
So, $\tilde{h}(Y^n)=-\sum_{\mathbf{i}}\tilde{f}_{\mathbf{i}}$. Using the same method as in the previous step, we have:
\begin{align*}
 &\mathbb{E}_{X^n,Y^n}(\tilde{e}\tilde{h})=-\mathbb{E}_{X^n,Y^n}(\tilde{e}\tilde{f})\leq \sum_{\mathbf{i}}C_{\mathbf{i}}\mathbf{P}^{\frac{1}{2}}_{\mathbf{i}}\mathbf{Q}^{\frac{1}{2}}_{\mathbf{i}} 
 \\&\Rightarrow  P(e(X^n)\neq h(Y^n)) \geq2\sqrt{\sum_{\mathbf{i}}\mathbf{P}_{\mathbf{i}}}\sqrt{\sum_{\mathbf{i}}\mathbf{Q}_{\mathbf{i}}}-2\sum_{\mathbf{i}}C_\mathbf{i}\mathbf{P}_{\mathbf{i}}^{\frac{1}{2}}\mathbf{Q}_{\mathbf{i}}^{\frac{1}{2}} 
\end{align*}
On the other hand $P(e(X^n)\neq h(Y^n))=P(e(X^n)\neq 1\oplus f(Y^n))=P(e(X^n)= f(Y^n))=1-P(e(X^n)\neq f(Y^n)$. So, 
\begin{align*}
&1-P(e(X^n)\neq f(Y^n)) \geq 2\sqrt{\sum_{\mathbf{i}}\mathbf{P}_{\mathbf{i}}}\sqrt{\sum_{\mathbf{i}}\mathbf{Q}_{\mathbf{i}}}-2\sum_{\mathbf{i}}C_\mathbf{i}\mathbf{P}_{\mathbf{i}}^{\frac{1}{2}}\mathbf{Q}_{\mathbf{i}}^{\frac{1}{2}} 
\\&\Rightarrow 
P(e(X^n)\neq f(Y^n)) \leq 1- 2\sqrt{\sum_{\mathbf{i}}\mathbf{P}_{\mathbf{i}}}\sqrt{\sum_{\mathbf{i}}\mathbf{Q}_{\mathbf{i}}}+2\sum_{\mathbf{i}}C_\mathbf{i}\mathbf{P}_{\mathbf{i}}^{\frac{1}{2}}\mathbf{Q}_{\mathbf{i}}^{\frac{1}{2}} 
.
\end{align*}
This completes the proof.
\end{IEEEproof}


\begin{thebibliography}{1}
%

\bibitem{ComInf1}
P.~Gacs and J.~K\"{o}rner, ``Common information is far less than mutual
  information,'' \emph{Problems of Control and Information Theory}, vol.~2,
  no.~2, pp. 119--162, 1972.

\bibitem{ComInf2}
H.~S. Witsenhausen, ``On sequences of pair of dependent random variables,''
  \emph{SIAM Journal of Applied Mathematics}, vol.~28, no.~1, pp. 100--113,
  1975.
\bibitem{FinLen}
F.~S. Chaharsooghi, A.~G. Sahebi, and S.~S. Pradhan, ``Distributed source
  coding in absence of common components,'' in \emph{Information Theory
  Proceedings (ISIT), 2013 IEEE International Symposium on}, July 2013, pp.
  1362--1366.

\bibitem{security2}
A.~Bogdanov and E.~Mossel, ``On extracting common random bits from correlated
  sources,'' \emph{IEEE Transactions on Information Theory}, vol.~57, no.~10,
  pp. 6351--6355, Oct 2011.

\bibitem{security3}
I.~Csiszar and P.~Narayan, ``Common randomness and secret key generation with a
  helper,'' \emph{IEEE Transactions on Information Theory}, vol.~46, no.~2, pp.
  344--366, Mar 2000.

\bibitem{control}
A.~Mahajan, A.~Nayyar, and D.~Teneketzis, ``Identifying tractable decentralized
  control problems on the basis of information structure,'' in \emph{2008 46th
  Annual Allerton Conference on Communication, Control, and Computing}, Sept
  2008, pp. 1440--1449.


\bibitem{Reed_and_Simon}
M.~Reed and B.~Simon, \emph{Methods of Modern Mathematical Physics, I:
  Functional Analysis}.\hskip 1em plus 0.5em minus 0.4em\relax New York:
  Academic Press Inc. Ltd., 1972.
\bibitem{arxiv2}
F. Shirani, M. Heidari, S. S. Pradhan, \emph{On the Sub-optimality of Single-letter Coding in Multi-letter Communications}, arxiv.org, Jan 2017.
\bibitem{arxiv1}
F. Shirani, S. S. Pradhan, \emph{On the Correlation between Boolean Functions of Sequences of Random Variables}, arxiv.org, Jan 2017.
\end{thebibliography}

\end{document}